\documentclass[english, 11pt]{article}
\usepackage{microtype}
\usepackage{amsmath}
\usepackage{amssymb}
\usepackage{amsthm}
\usepackage{varioref}
\usepackage{microtype}
\usepackage{verbatim} 
\usepackage{multicol}
\usepackage{lmodern}
\usepackage{enumerate}
\usepackage[normalem]{ulem}
\usepackage[margin=1in]{geometry}
\usepackage{caption}
\usepackage{subcaption}
\usepackage[T1]{fontenc}

\usepackage{mathrsfs}

\usepackage{url}
\usepackage{xspace}
\usepackage{thm-restate}

\usepackage{hyperref}
\hypersetup{
    bookmarksnumbered=true, 
    unicode=false, 
    pdfstartview={}, 
    pdftitle={}, 
    pdfauthor={}, 
    pdfsubject={}, 
    pdfcreator={}, 
    pdfproducer={}, 
    pdfkeywords={}, 
    pdfnewwindow=true, 
    colorlinks=true, 
    linkcolor=blue, 
    citecolor=blue, 
    filecolor=blue, 
    urlcolor=blue 
}

\newtheorem{theorem}{Theorem}
\newtheorem{lemma}[theorem]{Lemma}

\newtheorem{definition}[theorem]{Definition}

\newtheorem{conjecture}[theorem]{Conjecture}
\newtheorem*{conjecture*}{Conjecture}

\theoremstyle{definition}

\newtheorem{recipe}{Recipe}

\newtheorem{innercustomthm}{Theorem}

\newtheorem{innercustomlemma}{Lemma}
\newenvironment{customlemma}[1]
  {\renewcommand\theinnercustomlemma{#1}\innercustomlemma}
  {\endinnercustomthm}

\theoremstyle{definition}

\theoremstyle{remark}

\newcommand\BQEXP{{\sf{BQEXP}}}
\newcommand\QMA{{\sf{QMA}}}
\newcommand\QCMA{{\sf{QCMA}}}
\newcommand\PP{\sf{PP}}
\newcommand\expQCMA{{\sf{QCMA_{exp,poly}}}}
\newcommand\NP{{\sf{NP}}}
\newcommand\AM{{\sf{AM}}}

\newcommand{\eq}[1]{(\hyperref[eq:#1]{\ref*{eq:#1}})}

\newcommand{\ket}[1]{|#1\rangle}

\newcommand{\tth}[0]{\textsuperscript{th}}

\usepackage{xcolor}

\newcommand{\wf}[1]{{\color{violet}{[{\bf wf:}#1]}}}

\DeclareMathAlphabet{\matheu}{U}{eus}{m}{n}

\DeclareMathOperator{\tr}{tr}

\newcommand{\sop}[1]{{\mathcal #1}}

\newcommand{\I}{{\mathbb I}}

\newcommand{\poly}{\textrm{poly}}

\newcommand{\ketbra}[2]{|{#1}\rangle\!\langle{#2}|}

\newcommand{\even}{_\textrm{even}}
\newcommand{\odd}{_\textrm{odd}}

\newcommand{\orcl}{{\pmb{\sop O}}}
\newcommand{\SA}{{\mathbf{S}_X^n}}
\newcommand{\SB}{{\mathbf{S}_Y^n}}

\usepackage{mathtools}
\newcommand{\pc}{preimage-correct}
\newcommand{\rpc}{randomized-preimage-correct}
\newcommand{\spi}{\pmb{\sigma}_{\textrm{pre}}}
\newcommand{\sperm}{S_{\textrm{pre}}}
\newcommand{\Swit}{\mathbf{S}_{i,\textrm{wit}}}
\newcommand{\Sprime}{\mathbf{S}''}
\newcommand{\porcl}{{\pmb{\sop P}}}

\usepackage{authblk}

\begin{document}

\title{Quantum vs. Classical Proofs and Subset Verification}
\author[1]{Bill Fefferman}
\affil[1]{Department of EECS, University of California at Berkeley, Berkeley CA, USA, \texttt{wjf@berkeley.edu}}
\author[2]{Shelby Kimmel}
\affil[2]{Department of Computer Science, Middlebury College, Middlebury, VT, USA, \texttt{{skimmel@middlebury.edu}}}

\date{}
\maketitle
\begin{abstract}
We study the ability of efficient quantum verifiers to decide properties of exponentially large subsets given either a classical or quantum
 witness. We develop a general framework that can be used to prove that
$\QCMA$ machines, with only classical witnesses, cannot verify certain properties of subsets given implicitly via an oracle. We use this framework
to prove an oracle separation between $\QCMA$ and $\QMA$ using
an ``in-place'' permutation oracle, making the first progress on this question since Aaronson and Kuperberg in 2007 \cite{AK07}.  We also use the framework
to prove a particularly simple standard oracle separation between $\QCMA$
and $\AM.$ 
\end{abstract}
 
\section{Introduction}

How much computational power does an efficient quantum verifier gain 
when given a polynomial sized quantum state to support the validity of a
mathematical claim?  In particular, is there a problem that can be
solved in this model, that cannot be solved if the
verifier were instead given a classical bitstring?  This question, the so-called
$\QMA$ vs. $\QCMA$ problem, is fundamental in quantum
complexity theory.  To complexity theorists, the question can be
motivated simply by trying to understand the power of quantum
nondeterminism, where both $\QMA$ and $\QCMA$ can be seen as ``quantum analogues''
of $\NP$.  More physically, $\QMA$ is characterized by the $k$-local
Hamiltonian problem, in which we must decide if the ground
state energy of a local Hamiltonian is above or below a specified
threshold \cite{KSV02,AN02}.  In this setting, the $\QMA$ vs. $\QCMA$
question asks whether there exists a
purely {\sl classical} description of the ground state that allows us
to make this decision.  For instance, if the ground state of any local
Hamiltonian can be prepared by an efficient quantum circuit, then
$\QMA=\QCMA$, as the classical witness for the $k$-local
Hamiltonian problem would be the classical description of this quantum
circuit.  It was this intuition that caused Aharonov and Naveh
to conjecture that these classes are equal, in the paper that first
defined $\QCMA$ \cite{AN02}.

It was recently established \cite{GKS} that the witness to a $\QMA$ machine may
always be replaced by a subset state, where a subset state
on $n$ qubits has the form $\ket{S}=1/\sqrt{|S|}\sum_{i\in
S}|i\rangle$ for some subset $S\subset[2^n]$.  However, 
it seems difficult to create a classical witness on $n$
bits that captures the information in a subset state $\ket{S}.$
Therefore, problems involving subsets seem like ripe ground for
understanding the $\QMA$ vs. $\QCMA$ problem. We investigate
the following question: under what circumstances is it possible
for a quantum machine to verify properties of a subset? This
question is not answered by \cite{GKS}; they study general
properties of languages that are in $\QMA$ and $\QCMA$, while
we attempt to prove specific languages of interest (that are
related to subsets) are either in or not in $\QMA$
or $\QCMA$.

In the hopes of further exploring these questions, we 
exhibit a general framework that can be used to obtain oracle
separations against $\QCMA$ for subset-based problems.
We use this  framework to prove the existence of an ``in-place''
permutation oracle $\porcl$ (a unitary that permutes standard basis
states within a single register) \cite{dBCW02,collision} for which
$\QMA^\porcl\not\subseteq \QCMA^\porcl$, making the first progress on
this problem since Aaronson and Kuperberg in 2007 \cite{AK07}, who attained a ``quantum oracle'' separation (i.e., a separation relative to an \emph{arbitrary} black-box unitary transformation acting on a polynomial number of qubits). In this
problem, for the case of $\QMA$,
the in-place permutation oracle allows us to verify that the given witness
is indeed the correct subset state. On the other hand, our framework
allows us to prove the language is not in $\QCMA.$  Our framework 
is quite general, and we are also able
to use it to  establish a particularly simple example of a
(conventional) oracle $\orcl$ so that
$\AM^{\orcl}\not\subseteq\QCMA^{\orcl}$.\footnote{Note there was
previously an example of an oracle separating $\AM$ from $\PP$
\cite{V95}.  Since $\QMA\subseteq\PP$ \cite{MW05}, this is formally a
stronger result.  Nonetheless, our oracle is substantially different,
and uses completely different ideas.}

\subsection{Subset-Verifying Oracle Problems}

We consider two oracle problems related to verifying
properties of subsets. In {\it{Subset Size
Checking}}, we are given a black box function
$f:[N^2]\rightarrow\{0,1\}$, that marks elements with either a 0 or 1.
We are promised that the number of marked items is either 
$N$ or $0.99N$, and we would like to decide which is the case. We want to verify the size of the subset marked by $f.$

In the other oracle problem, {\it{Preimage Checking}}, we are given a
black box permutation on $N^2$ elements. We are promised that the 
preimage of the first $N$ elements under the permutation is either mostly
even or mostly odd, and we would like to decide which is the case. In
this problem, we want to verify this parity property of a subset of the 
preimage of the function.

{\it{Subset Size Checking}}  is in {\AM} \cite{GS86}, and
we give a procedure that proves {\it{Preimage
Checking}} is in {\QMA} when the permutation is given as an in-place 
quantum oracle. An in-place permutation unitary $\sop
P_\sigma$ acts as $\sop P_\sigma\ket{i}=\ket{\sigma(i)}$ for a
permutation $\sigma$. For {\it{Preimage Checking}}, we are interested
in the set $S_{\textrm{pre}}(\sigma)=\{i:\sigma(i)\in[N]\}$. Given the subset state
$\ket{S_{\textrm{pre}}(\sigma)}$, it is easy to verify that the correct state
was sent, because $\sop P_\sigma\ket{S_{\textrm{pre}}(\sigma)}=\ket{[N]},$
which is easy to verify using a measurement in the Hadamard basis.

However, we do not expect subset-based oracle languages like
{\it{Subset Size Checking}} and {\it{Preimage
Checking}} to be in $\QCMA$ because the classical witness does not have
enough information to identify the relevant subset.
We make this intuition more precise by providing a general recipe for
proving that subset-verifying oracle languages are not in {\QCMA}. We
apply this procedure to show that both
{\it{Preimage Checking}} (with a randomized  in-place oracle - see Section
\ref{Sec:inplace} for more details) and {\it{Subset Size Checking}} 
are not in {\QCMA}. The procedure
involves familiar tools, like the adversary bound \cite{A00} (although
adapted to our in-place oracle when necessary), as well as
a new tool, the {\it{Fixing Procedure}}, which finds
subsets with nice structure within a large arbitrary set. We
now sketch the recipe: 

\begin{enumerate}

\item  We show that for every {\QCMA} machine, there are more
valid oracles than possible classical witnesses, so by
a counting argument, there must be one classical
witness $w^*$ that corresponds to a large number of potential oracles.
We then restrict ourselves to considering
oracles that correspond to $w^*.$

\item Because we are considering subset-verifying problems, 
if we have a collection of black box functions that corresponds to $w^*$,
we immediately have some set of subsets that corresponds to $w^*.$
At this point, we know nothing about this set of subsets
except its size, thanks to the counting argument. We next show
(using the {\it{Fixing Procedure}})
that if we have a set of subsets of a certain size, we can always
find a subset of the original set that has nice structure. 

\item We apply the adversary bound to the subset with
nice structure to show that the number of quantum queries needed
to distinguish between YES and NO cases is exponential.

\item We finally put these pieces together in a standard diagonalization argument.

\end{enumerate}

\subsection{Technical Contributions}\label{sec:technical}

Our adversary bound for in-place permutation oracles provides a query
lower-bounding technique for unitary oracles when access is {\it not} given 
to the inverse of the oracle.
(While Belovs \cite{B15} created an adversary bound for arbitrary
unitaries, his results assume access to an inverse.) While we typically 
assume quantum oracles include access to an inverse or are self-inverting, 
in open quantum systems it is natural to not have an inverse.

When proving that {\it{Preimage Checking}} is not in {\QCMA}, we
 use an oracle that is not unitary. The oracle is a completely-positive trace-preserving (CPTP)  map that at each application applies
one unitary  chosen uniformly at random from among a set of unitaries.
Standard lower bounding techniques fail for such an oracle.  The
closest result is from Regev and Schiff \cite{regev2008impossibility} who give a lower
bound on solving Grover's problem with an oracle that produces errors.
Regev and Schiff deal with the  non-unitarity of the map by modeling
the state of the system using pure states. This
strategy does not work in our case. Instead, we use the
fact that every non-unitary CPTP map can be implemented as a unitary
on a larger system.  In our case, we simulate our random oracle
using a unitary black box oracle acting on subsystem $B$, followed
by a fixed unitary that entangles subsystems $A$ and $B$. 
The entangling operation can
not be efficiently implemented, but as we are bounding query
complexity, this is acceptable.  This technique may be of use for
similar problems; for example, we do not know the query complexity of
solving Grover's problem with an oracle that produces a depolarizing
error with each application. A depolarizing map is similar to our CPTP
map in that both maps can be thought of as applying a unitary at
random from among a set of unitaries, and so perhaps this approach
will stimulate new approaches for the Grover problem.

\subsection{Impact and Directions for Future Research}

While Aaronson and Kuperberg have previously proved an oracle separation between
$\QMA$ and $\QCMA$ \cite{AK07}, their oracle seems to be especially quantum,
as it is defined by a Haar random quantum state. Our in-place oracle has
more of a classical feel, in that it encodes a classical permutation
function. However, it is still not a standard quantum oracle, as it is
not self-inverting.
 Is there a standard (i.e. not in-place)
oracle  language that separates {\QMA} and {\QCMA}? 
Although
we can only prove a separation when our in-place oracle
also has randomness, we believe our techniques could be 
adapted to prove a similar result but without oracle randomness. 
While we give a recipe for showing certain subset-based 
problems are not in $\QCMA$, we believe some of
these problems are also not in $\QMA$; for example, is it possible
to prove {\it{Subset Size Checking}} is not in $\QMA$?

Our contributions to techniques for lower bounding query
complexity for non-standard oracles raise several questions. Is there
a general adversary bound \cite{HLS07,LMR+11} for in-place permutation oracles? There are
examples of problems for which in-place permutation oracles require
exponentially fewer queries that standard permutation oracles e.g.,
\cite{AMR+11}. We conjecture that
there are also examples of problems for which standard permutation oracles
require exponentially fewer queries than in-place oracles. In fact, we
do not believe it is possible to obtain a Grover-type speed-up with an
in-place oracle; we believe the problem of determining the inverse of
an element of an in-place permutation oracle requires $N$ queries for
a permutation on $N$ elements. However, in order to prove these
results, we suspect one needs a more powerful tool, like a general
adversary bound for in-place oracles.

\subsection{Organization}
The rest of the paper is organized as follows, in Section
\ref{sec:notation}, we introduce notation that will be used throughout
the rest of the paper, and define $\QMA$ and $\QCMA$. In Section 
\ref{Sec:inplace}, we define and discuss standard,  in-place, and randomized in-place quantum
permutations, as well as state an adversary lower bound
for in-place permutation unitaries. In Section \ref{sec:QMA}, we
describe the {\it{Preimage Checking}} Problem and prove it is in $\QMA$. In 
Section \ref{sec:conditions}, we lay out the general recipe for
proving subset-based languages are not in {\QCMA}. In Section 
\ref{sec:counting}, we apply this procedure to the {\it{Subset Size Checking}}
problem, 
and use it to prove an oracle separation between {\AM} and {\QCMA}. 
Finally, in Section \ref{sec:QCMA}, we apply the procedure to {\it{Preimage 
Checking}} and show this problem is not in {\QCMA}. 

\section{Definitions and Notation}\label{sec:notation}

We use the notation $[M]=\{1,2,\dots, M\}$.  $\sigma$ will refer to a
permutation.  Unless otherwise specified, the sets we use, generally denoted $S$, will be a set of positive integers. Also, we will use bold type-face to denote a set of sets.  For instance, $\mathbf{S}$ will refer to a set of sets of positive integers.  To make our notation clearer, we will refer to such a set of sets as a \emph{set family}.  Likewise, we denote $\pmb{\sigma}$ to be a set of permutations acting on the same set of elements.

For $S$ a set of positive integers, a subset
state $\ket{S}$ is $\ket{S}=\frac{1}{\sqrt{|S|}}\sum_{i\in S}\ket{i}.$

Throughout, we use $N=2^n$. All logarithms are in base 2.
We use $\pmb{\sigma}^n$ to be the set of permutations acting on $N^2$ elements.
That is, if $\sigma\in \pmb{\sigma}^n$, $\sigma:[N^2]\rightarrow[N^2]$.
For positive integers $i>j$, let $\mathbf{C}(i,j)$ be the set family
containing $j$ elements of $[i]:$ $\mathbf{C}(i,j)=\{S\subset[i]:|S|=j\}.$

We use calligraphic font $\sop P$, $\sop U$ to denote unitary operations. We use
elaborated calligraphic font $\mathscr P$, $\mathscr U$ to denote CPTP maps.
For a unitary CPTP map $\mathscr U$ acting on a density matrix $\rho$, we have
$\mathscr U(\rho)=\sop U\rho\sop U^\dagger,$
where $(\cdot)^\dagger$ denotes conjugate transpose. We will use $\sop O$ to denote a unitary oracle, and $\mathscr O$ to denote a CPTP map oracle.

We include the following standard definition for completeness (e.g., see also
\cite{complexityzoo,AK07}).

\begin{definition}$\QMA$ is the set of promise problems $A=(A_{Yes},A_{No})$ so that there exists an efficient quantum verifier $V_{A}$ and a polynomial $p(\cdot)$: \label{def:qma}
\begin{enumerate}
\item{\emph{Completeness}: For all $x\in A_{Yes}$ there exists a $p(|x|)$-qubit pure quantum state $|\psi\rangle$ so that $\Pr\left[V_A(x,|\psi\rangle)=1\right]\geq 2/3$}
\item{\emph{Soundness}: For all $x\in A_{No}$ and any pure quantum state $|\psi\rangle$, $\Pr\left[V_A(x,|\psi\rangle\right)=1]\leq 1/3$.}    
\end{enumerate}
$\QCMA$ is the same class, with the witness $|\psi\rangle$ replaced by a polynomial length classical bitstring.  \end{definition}


\section{Permutation Maps}\label{Sec:inplace}

\subsection{Permutations as Oracles: In-Place Permutation vs. Standard Permutation}
Black box permutation unitaries have been considered previously, most notably in the collision
and element distinctness problems \cite{collision,AS2004}. However,
the unitaries considered in these works were standard
oracles.  A standard oracle implements
the permutation $\sigma\in\pmb{\sigma}^n$ as
$\sop P^{\textrm{stand}}_{\sigma}\ket{i}\ket{b}=\ket{i}\ket{b\oplus\sigma(i)}$
for $i,b\in[N^2]$, where $\ket{i}$ for $i\in[N^2]$ are standard basis
states and $\oplus$ denotes bitwise XOR. Note that $(\sop P^{\textrm{stand}}_{\sigma})^2=\I_{N^4}$; 
that is, acting with the
unitary twice produces the $N^4\times N^4$ identity operation.

We consider in-place permutation unitaries, which
implement the permutation $\sigma\in\pmb{\sigma}^n$ as
$\sop P_{\sigma}\ket{i}=\ket{\sigma(i)}.$
In general $(\sop P_{\sigma})^2\neq \I_{N^2}$. Crucially, given
black box access to $\sop P_{\sigma}$,
{\it{we do not give black box access to its inverse.}} In fact,
in Section \ref{sub:ambainis}, we show that given only $\sop P_{\sigma}$, it is hard to
invert its action. Non-self-inverting permutation unitaries have been 
considered previously, in \cite{dBCW02,collision}.

We believe standard and in-place permutation unitaries 
are of incomparable computational power. That is, given one type that implements
$\sigma$, you can not efficiently simulate the other type
implementing the permutation $\sigma$.  For example, if we have
the state $\sum_{y\in S}|y\rangle|\sigma(y)\rangle$ (normalization omitted), we
can create the state $\sum_{y\in S}|y\rangle|0\rangle$ with a
single query to $\sop P_{\sigma}^{\textrm{stand}}$.  However, if we only
have access to the in-place permutation $\sop
P_{\sigma}$ and not to $\sop
P_{\left(\sigma\right)^{-1}}=\left(\sop
P_{\sigma}\right)^{-1}$, it seems difficult to create this state.

On the other hand, suppose we want to prepare the state $\sum_{y\in
[N]}|\sigma(y)\rangle$ (normalization omitted).  We can create this state in one query
to the in-place permutation oracle $\sop P_{\sigma}$
by applying the oracle to the uniform superposition $\sum_{y\in
[N]}|y\rangle$. In the standard model,
this problem is called ``index erasure,'' and requires an exponential
number of queries in $n$ to $\sop P_{\sigma}^{\textrm{stand}}$ \cite{AMR+11}.

\subsection{An Adversary Bound for In-Place Permutation Oracles}\label{sub:ambainis}

In Appendix \ref{sec:ambainis}, we prove a non-weighted adversary
bound for in-place permutations oracles that is identical to what Ambainis proves in Theorem 6 in \cite{A00}
for standard permutation oracles. 
\begin{restatable}{lemma}{amb}\label{lemm:ambainis}
Let $\pmb{\sigma}\subset[V]\rightarrow[V]$ be a subset of permutations
acting on the elements $[V]$. Let $f:\pmb{\sigma}\rightarrow\{0,1\}$
be a  function of permutations. Let
$\pmb{\sigma}_X\subset \pmb{\sigma}$ be a set of permutations such that if $\sigma\in
\pmb{\sigma}_X$, then $f(\sigma)=1$. Let
$\pmb{\sigma}_Y\subset \pmb{\sigma}$ be a permutation family such that if $\sigma\in \pmb{\sigma}_Y$ 
then $f(\sigma)=0$. Let $R\subset \pmb{\sigma}_X\times \pmb{\sigma}_Y$ be such that
\begin{itemize}
\item For every $\sigma_x\in \pmb{\sigma}_X$, there exists at least $m$ 
different $\sigma_y\in \pmb{\sigma}_Y$ such that
$(\sigma_x,\sigma_y)\in R$.
\item For every $\sigma_y\in \pmb{\sigma}_Y$, there exists at least $m'$ 
different $\sigma_x\in \pmb{\sigma}_X$ such that
$(\sigma_x,\sigma_y)\in R$.
\item 
Let $l_{x,i}$ be the number of $\sigma_y\in \pmb{\sigma}_Y$ such that
$(\sigma_x,\sigma_y)\in R$ and $\sigma_x(i)\neq \sigma_y(i).$ Let
$l_{y,i}$ be the number of $\sigma_x\in \pmb{\sigma}_Y$ such that
$(\sigma_x,\sigma_y)\in R$ and $\sigma_x(i)\neq \sigma_y(i).$ Then let
$l_{max}=\max_{(\sigma_x,\sigma_y)\in R,i}l_{x,i}l_{y,i}.$
\end{itemize}
Then given an in-place permutation oracle $\sop P_{\sigma}$ for $\sigma\in\pmb{\sigma}$ that acts as
$\sop P_{\sigma}\ket{i}=\ket{\sigma(i)},$
any quantum algorithm
that correctly evaluates $f(\sigma)$ with probability $1-\epsilon$ for every element
of $\pmb{\sigma}_X$ and $\pmb{\sigma}_Y$
must use $\left(1-2\sqrt{\epsilon(1-\epsilon)}\right)\sqrt{\frac{mm'}{l_{max}}}$ queries
to the oracle.
\end{restatable}

As a corollary of Lemma \ref{lemm:ambainis}, (using exactly the same
technique as Theorem 7 in \cite{A00}), we have that the query
complexity of inverting an in-place permutation oracle on $V$ elements is $\Omega(V^{1/2}).$

\subsection{Permutations with Randomness}

Additionally, we consider in-place permutation oracles with internal
randomness that are CPTP (completely-positive trace-preserving)
maps, rather than unitaries.
Oracles with internal randomness were shown to  cause a complete loss
of quantum speed-up in \cite{regev2008impossibility}, while in
\cite{harrow2014uselessness}, such oracles were shown to give an
infinite quantum speed-up.

We consider oracles that apply an in-place permutation
at random from among a family of possible permutations. Let $\pmb{\sigma}\subseteq\pmb{\sigma}^n$
be a set of $|\pmb{\sigma}|$ permutations. 
Then the CPTP map $\mathscr P_{\pmb{\sigma}}$ acts
as follows:
\begin{align}
\mathscr P_{\pmb{\sigma}}(\rho)=
\frac{1}{|\pmb{\sigma}|}\sum_{\sigma\in\pmb{\sigma}}
\sop P_\sigma\rho\sop P_\sigma^\dagger.
\end{align}


\section{Pre-Image Checking}\label{sec:QMA}

In this section, we define a property of oracle families which we call
{\it{\rpc}}, and construct a decision language based on such
oracles that is in $\QMA$.  Essentially, the problem is to decide
whether the preimage of the first $N$ elements of a permutation
is mostly even or odd. 
 
Given a permutation $\sigma\in\pmb{\sigma}^n$, we associate a preimage subset
$\sperm(\sigma)$ to that permutation (``pre'' is for ``preimage''), where
$\sperm(\sigma)=\{j:\sigma(j)\in[N]\}.$
That is, $\sperm(\sigma)$ is the subset of elements in $[N^2]$ whose image under
$\sigma$ is in $[N]$.
Additionally, to each subset $S\subseteq[N^2]$ with $|S|=N$, we associate a set
of permutations $\spi(S)$, where $\spi(S)=\{\sigma:\sigma\in\pmb{\sigma}^n, \sperm(\sigma)=S\}.$
Let
\begin{align}
\pmb{S}^n\even&=\{S:S\subset[N^2],|S|=N,|S\cap\mathbb{Z}\even|=2/3 N\}\nonumber\\
\pmb{S}^n\odd&=\{S:S\subset[N^2],|S|=N,|S\cap\mathbb{Z}\odd|=2/3 N\}.
\end{align}


\begin{definition}[{\rpc} oracles]\label{def:rpc}
Let $\orcl$ be a countably infinite set of quantum operators (CPTP maps): 
$\orcl=\{\mathscr O_1,\mathscr O_2,\dots\}$, 
where each $\mathscr O_n$ implements an operation on $(2n)$-qubits. 
We say that $\orcl$ is {\rpc}
if for every $n$, $\mathscr O_n=\mathscr P_{\pmb{\sigma}_\emph{pre}(S)}$, with 
$S\in\pmb{S}^n\even\cup \pmb{S}^n\odd$.
\end{definition}

\begin{theorem}\label{thm:qma}
For any {\rpc} $\orcl$, the unary
language $L_{\orcl}$, which contains those unary strings $1^n$ such that
$\mathscr O_n=\mathscr P_{\pmb{\sigma}_\emph{pre}(S)}$ with $S\in\pmb{S}^n\even$, is in
$\QMA^{\orcl}$.
\end{theorem}
\begin{proof}

We first prove completeness. We assume $1^n\in L_{\orcl}$, so $\mathscr
O_n=\mathscr P_{\spi(S)}$ for some $S\in \pmb{S}^n\even.$
We consider using as a witness the subset state $\ket{S}$ on $2n$
qubits. 
We analyze the following verifier: with probability
$1/2$, do either
\begin{enumerate}[{\bf{Test}} (i)]

\item Apply $\mathscr P_{\spi(S)}$ to $\ket{S}$, and
measure whether the resultant state is
$\ket{[N]}.$ This measurement can be done by applying $H^{\otimes n}$ to the first
$n$ qubits, and then measuring all qubits in the standard
basis.  If the outcome is 0, output 1;
otherwise, output 0.

\item Measure $\ket{S}$ in the standard basis. Let $i^*$ be the resulting
standard basis state. If $i^*$ is odd, output 0. Otherwise, 
apply $\mathscr P_{\spi(S)}$ to $\ket{i^*}$ and measure the resultant
standard basis state. If the resultant
state is not in $[N]$, output 0; otherwise, output 1. 

\end{enumerate} 
If Test (i) is implemented, the verifier always outputs 1 because
all the permutations that might be applied by $\mathscr P_{\spi(S)}$ transform $\ket{S}$
into $\ket{[N]}$. If Test (ii) is implemented, 
the verifier outputs 1 with probability 2/3. Averaging over both Tests, the verifier
 outputs 1 with probability 5/6.

Now we show soundness. Let $1^n\notin L_{\orcl}$, so $\mathscr O_n=\mathscr
P_{\spi(S)}$ for some $S\in\mathbf{S}^n\odd$. 
%
Without loss of generality, let the witness be the $2n$-qubit state
$\ket{\psi(S)}=\sum_{i=1}^{N^2}\beta_i\ket{i}$.
If $p_{(i)}$ (resp. $p_{(ii)}$) is the
 probability the verifier outputs 1 after performing Test (i) (resp.
 Test (ii)), then we have
\begin{align}
p_{(i)}=\frac{1}{N}\left|\sum_{i\in S}\beta_i\right|^2,
&&p_{(ii)}=\sum_{i\in \mathbb Z\even \cap S}|\beta_i|^2.
\end{align}
regardless of which permutation the map $\mathscr P_{\spi(S)}$ applies. 

Using Cauchy-Schwarz and the triangle inequality, we have
$1\geq\left(\sqrt{3p_{(i)}}+
(\sqrt{2}-1)p_{(ii)}\right)/\sqrt{2}.$
Thus the total probability that the verifier outputs 1 is
\begin{align}
\frac{1}{2}\left(p_{(i)}+p_{(ii)}\right)\leq 
\frac{1}{2}\left(\frac{2}{3}\left(1-\frac{\sqrt{2}-1}{\sqrt{2}}p_{(ii)}\right)^2+p_{(ii)}\right).
\end{align}
The derivative of the right hand side is positive
for $0\leq p_{(ii)}\leq 1$, so to maximize the right hand side we take
$p_{(ii)}=1.$ Doing this, we find the probability that the verifier outputs 1 is
at most $2/3.$
\end{proof}

We will show that the Preimage Checking problem is not in $\QCMA$ in Section \ref{sec:QCMA}.

Our proof that the Preimage Checking problem is in $\QMA$ works equally well
for an in-place oracle without randomness. We use the randomness in our oracle
in the proof that {\rpc} languages can not be decided by $\QCMA.$ We believe the separation holds
even without randomness in the oracle.

\section{Strategy for Proving Subset-Based Oracle Languages are not in $\QCMA$}
\label{sec:conditions}

In this section, we describe a general strategy for showing that certain
oracle languages are not in $\QCMA$. In particular, we consider the case when
the oracles are related to sets of integers:
\begin{definition}[Subset-Based Oracle]\label{def:subsetLanguage}
Let $\orcl=\{ \mathscr O_1,\mathscr O_2,\dots \}$ be an oracle such that each
$\mathscr O_n$ implements a $p_1(n)$-qubit CPTP map from some set of maps
$\orcl^n$. Then we say $\orcl$ is a subset-based oracle if there exists a
set of bijective functions $\{g^1,g^2,\dots\}$ with
$g^n:\orcl^n\rightarrow \mathbf{S}^n$ where $\mathbf{S}^n$ is the union of disjoint subset families $\SA$ and
$\SB$.


\end{definition}


We also use the following definition:
\begin{definition}\label{def:distr}
Given a subset family $\mathbf{S}$ containing subsets of positive integers, and $\beta\in \mathbb{R}$ such that $\beta>0$, we say $\mathbf{S}$ is $\beta$-distributed if:
\begin{enumerate}[(1)]
\item There exists a (possibly empty) set $S_{\emph{fixed}}$ such that 
$S_{\emph{fixed}}\subset S $ for all $ S\in \mathbf{S}$.
\item For every element $i\in\left(\bigcup_{S\in\mathbf{S}}S\right)\setminus S_{\emph{fixed}}$, 
$i$ appears in at 
most a $2^{-\beta}$-fraction of $S\in \mathbf{S}$.
\end{enumerate}
\end{definition}
\noindent We call $S_{\textrm{fixed}}$ the ``fixed subset'' of $\mathbf{S}.$

We use the following Recipe for proving a subset-based oracle language is not in $\QCMA$:
\vspace{.2cm}
\begin{recipe}\label{Recipe}
~\\
\begin{enumerate}
\item [Set-up:] Fix some enumeration
over all $\poly(n)$-size quantum verifiers
$M_1,M_2,...$,  which we can do because the number of such machines is
countably infinite (by the Solovay-Kitaev theorem
\cite{solovaykitaev}).  Some of these verifiers
may try to decide a language by trivially ``hardwiring'' its outputs; for example, by returning $1$
independent of the input. We start by fixing a unary language $L$ such that
no machine $M_i$ hardwires the language. We can always do this because
there are more languages than $\poly(n)$-sized machines. Then
our goal is to associate a subset-based oracle $\orcl=\{\mathscr O_1,\mathscr O_2,\dots\}$ with $L$, such that $1^n\in L$ if and only if $g^n(\mathscr O_n)\in \SA$, and to show that even with access to $\orcl$, no $M_i$ can efficiently decide $L$ for all $n.$ 

Consider the $\QCMA$ machine $M_i$, and suppose it is given access to a subset-based oracle $\orcl$, as well as a witness of $p_{M_i}(n)$ bits for each input $1^n$. Then for each $\mathscr O_n\in\orcl$ there is some subset of integers $S\in \mathbf{S}^n$ such that $g^n(\mathscr O_n)=S$. Since $g^n$ is bijective, $S$ uniquely defines $\mathscr O_n$, so the optimal witness that causes $M_i$ to accept $\mathscr O_n$ can be thought of as a function of $S.$ Let $w_i(S)$ be the witness that
gives the highest probability of success in convincing $M_i$ that $S\in \SA$.
Then we denote $\Swit(w)=\{S:S\in
\SA, w=w_i(S)\}.$

Using the pigeonhole principle, there exists some string $w_{i,n}$
of $p_{M_i}(n)$ bits such that
\begin{align} \label{eq:def_Sw}
|\Swit(w_{i,n})|\geq \frac{1}{2^{p_{M_i}(n)}}\left|\SA\right|.
\end{align}
That is, there exists a witness such that a large number
of subsets correspond to that witness.

\item Prove that for $n\geq n_i^*$, there is a subset family $\mathbf{S}_X\subseteq\Swit(w_{i,n})$ that is $\alpha$-distributed with fixed subset $S_{\textrm{fixed}}$. Let $\mathbf{S}_Y=\{S:S\in \SB, S_\textrm{fixed}\subset S\}$. Show the cardinality of $\mathbf{S}_Y$ is large. 

\item Create a relation $\mathbf{R}\subseteq
\{\mathscr O: \mathscr O\in \orcl^n, g(\mathscr O)\in \mathbf{S}_X\}
\times \{\mathscr O: \mathscr O\in \orcl^n, g(\mathscr O)\in \mathbf{S}_Y\}$ and use $\mathbf{R}$ to apply an
adversary bound to prove a lower bound of $\Omega(N^{\alpha/2})=\Omega(2^{n\alpha/2})$  on the number
of queries $M_i$ requires to distinguish some oracle $\mathscr O_{x,n,i}\in \orcl^{n}$ such that $g^n(\mathscr O_{x,n,i})\in \SA$ from an oracle $\mathscr O_{y,n,i}\in \orcl^{n}$ such that $g^n(\mathscr O_{y,n,i})\in \SB$.

\item Apply a standard Baker-Gill-Solovay diagonalization argument \cite{BGS} to complete the proof. That is, for each $M_i$, choose a unique $n_i\geq n_i^*$, and if $1^{n_i}\in L$, set $\mathscr O_{n_i}=\mathscr O_{x,n_i,i}$ and if $1^{n_i}\notin L$, set $\mathscr O_{n_i}=\mathscr O_{y,n_i,i}$. Then no $\QCMA$ machine can efficiently decide the language.
\end{enumerate}
\end{recipe}

\section{Subset Size Checking}\label{sec:counting}

In this section, we create a subset-based oracle language $L_\orcl$,
such that $L_\orcl\in \AM^\orcl$, but $L_\orcl\notin \QCMA^\orcl.$ We use
the strategy of Section \ref{sec:conditions} to prove $L_\orcl\notin \QCMA^\orcl$.

Let $f_S$ be a function that marks a subset $S\subset [N^2]$. That is $f_S:\{0,1\}^{2n}\rightarrow\{0,1\}$,
such that $f_S(i)=1$ if $i\in S$ and $0$ otherwise. Let $\sop F_S$ be the unitary such that 
$\sop F_S\ket{i}=(-1)^{f_S(i)}\ket{i}.$

\begin{definition}
Let $\orcl$ be a countably infinite set of unitaries (resp. boolean functions): $\orcl=\{\sop O_1,\sop O_2,...\}$,
where $\sop O_n$ implements a $2n$-qubit (resp. bit) unitary (function). We say $\orcl$ is subset-gapped if for 
every $n$, $\sop O_n=\sop F_S$ (resp. $\sop O_n=f_S$) for $|S|=N$ or $|S|=0.99N$. 
\end{definition}
\noindent Clearly $\orcl$ is a subset-based oracle (see Definition \ref{def:subsetLanguage}), with $g^n(\sop O_n)=g^n(\sop F_S)=S.$

Then the following two lemmas give the desired oracle separation between $\AM$ and $\QCMA$:

\begin{lemma}\label{lemma:gapped}
For any subset-gapped $\orcl$, the language $L_\orcl$ that contains those strings $1^n$ such that 
$\sop O_n=f_S$ with $|S|=N$, is in $AM^{\orcl}$.
\end{lemma}
\noindent Lemma \ref{lemma:gapped} is proven by Goldwasser and Sipser in \cite{GS86}.

\begin{lemma}\label{lemm:ACnotinQCMA}
For any subset-gapped $\orcl$, the language $L_\orcl$ that contains those strings $1^n$ such that 
$\sop O_n=\sop F_S$ with $|S|=N$, is not in $\QCMA^{\orcl}$.
\end{lemma}

To prove this lemma, we follow Recipe \ref{Recipe}. We address step 2 of the recipe in Lemma \ref{lemm:fixing_subset}:

\begin{restatable}{lemma}{fixingsubset}\label{lemm:fixing_subset}
Let $0<\alpha<1/2$ be a constant and $p(\cdot)$ be a polynomial
function.  Then there exists a positive integer $n^*(p,\alpha)$,
 such that for every positive integer $n>n^*(p,\alpha)$, and every subset family
$\mathbf{S}\subseteq\mathbf{C}(N^2,N)$ such that
$|\mathbf S|\geq|\mathbf{C}(N^2,N)|2^{-p(n)}$, there exists
a subset family $\mathbf{S}_X\subseteq \mathbf{S}$ such that $\mathbf{S}_X$ is 
$\alpha$-distributed with $|S_\textrm{fixed}|<.5N$. 
\end{restatable}
\noindent (Since $|S_\textrm{fixed}|<.5N$, this implies $|\{S:S\in \SB, S_\textrm{fixed}\subset S\}|$ is large, as desired.)

\noindent{\it Proof Sketch:} (Full proof in Appendix \ref{app:perm_proofs_subset}.)
We prove the existence of $\mathbf{S}_X$ by construction. Let
$\mathbf{S}$ be any subset of $\mathbf{C}(N^2,N)$ with
$|\mathbf{S}|\geq|\mathbf{C}(N^2,N)|2^{-p(n)}$. We construct
$\mathbf{S}_X$ using the Fixing Procedure:

\vspace{.3cm}
\fbox {
\parbox{.9\linewidth}{
\noindent {\bf{Fixing Procedure}}
\begin{enumerate}
\item Set $\mathbf{S}_X=\mathbf{S}$, and set $S_{\textrm{fixed}}=\emptyset$.
\item
\begin{enumerate}
\item Let $\nu(i)$ be the number of subsets $S\in {\mathbf{S}_X}$ such that
$i\in S$. 
\item If there
exists some element $i$ for which $ |{\mathbf{S}_X}|>\nu(i)\geq |{\mathbf{S}_X}|N^{-\alpha},$
 set $\mathbf{S}'\leftarrow\{S:S\in {\mathbf{S}_X}\text{ and } i\in S\},$ set 
 $S_\textrm{fixed}\leftarrow S_\textrm{fixed}\cup i,$ and return to step $2$(a).
 Otherwise exit the Fixing Procedure. 
\end{enumerate} 
\end{enumerate}
}
}
\vspace{.3cm}

By construction, the Fixing Procedure returns a set that is $\alpha$-distributed 
(see Definition \ref{def:distr}), so we only need to ensure that not too many elements are fixed. We obtain a lower bound on the final size of $\mathbf{S}_X$ because
each time an element is fixed, the size of the set decreases by at
most $N^{-\alpha}$. On the other
hand, because $\mathbf{S}_X$ is contained in $\mathbf{C}(N^2,N)$,
if a certain number of items are fixed, we have an upper bound on the
size of $\mathbf{S}_X$ using the structure of $\mathbf{C}(N^2,N)$
and a combinatorial argument. We show that if more than $.5N$ items
are fixed, these upper and lower bounds contradict each other, proving
that less than $.5N$ items must be fixed before the Fixing Procedure 
terminates.

We address Step 3 of Recipe \ref{Recipe} with the following Lemma:
\begin{restatable}{lemma}{conditionssubset}\label{lemm:X_conditions_subset}
Suppose $\mathbf{S}_X\subset \mathbf{C}(N^2,N)$ is the $\alpha$-distributed subset created using the Fixing Procedure of Lemma \ref{lemm:fixing_subset}, with fixed subset $S_{\textrm{fixed}}$. Let $\mathbf{S}_Y=\{S:S\in \mathbf{C}(N^2,0.99N), S_\textrm{fixed}\subset S\}$. Then we can construct an adversary bound to prove that 
for every quantum algorithm $G$, there exists $S_x\in\mathbf S_X$, and 
$S_y\in \mathbf{S}_Y$, (that depend on $G$) such that, given oracle access to $\sop F_{S_x}$ 
or $\sop F_{S_y}$, $G$ can not distinguish 
them with probability $\epsilon>.5$ without using $\left(1-2\sqrt{\epsilon(1-\epsilon)}\right)N^{\alpha/2}$ queries.
\end{restatable}

\noindent{\it Proof Sketch:} (Full proof in Appendix \ref{app:perm_proofs_subset}.)
We use Theorem 6 from \cite{A00}. This result is identical to our Lemma \ref{lemm:ambainis},
except with standard oracles rather than permutation oracles.

We let $\mathbf{R}= \mathbf{S}_X\times\mathbf{S}_Y$.
To apply Theorem 6, we need to show that for elements $i$ such that
$i\in S_x$ but $i\notin S_y$ for $(S_x,S_y)\in\mathbf{R}$ that either (1) $S_x$
is not connected to many other sets $S_y$ where $i\notin S_y$ or (2)
$S_y$ is not connected to many other sets $S_x$ where $i\in S_x$. We use
the $\alpha$-distributed property of $\mathbf{S}_X$ to show that property (2)
holds. We show a similar result for the case $i\notin S_x$ but $i\in S_y$ for $(S_x,S_y)\in\mathbf{R}$.

\section{Oracle Separation of $\QMA$ and $\QCMA$}\label{sec:QCMA}

In this section, we prove an oracle separation between $\QMA$ and $\QCMA$. 
In particular, we show:
\begin{theorem}\label{thm:main}
There exists a {\rpc} oracle $\orcl$, and a language $L_\orcl$ which contains those unary strings $1^n$ where
$\mathscr O_n=\mathscr P_{\pmb{\sigma}_\emph{pre}(S)}$ with $S\in\pmb{S}^n\even$
such that $L_\orcl\notin \QCMA^\orcl$.
\end{theorem}
\noindent Combined with Theorem \ref{thm:qma}, this gives the desired separation between $\QMA$ and $\QCMA$.

Really, we would like to prove a different result, one that involves preimage-correct oracles:
\begin{definition}[{\pc} oracles]
Let $\orcl$ be a countably infinite set of unitaries: 
$\orcl=\{\sop O_1,\sop O_2,\dots\}$, 
where each $\sop O_n$ implements an $(2n)$-qubit unitary. 
We say that $\orcl$ is {\pc}
if for every $n$, $\sop O_n=\sop P_{\sigma}$, for some
$\sigma$ such that $\sperm(\sigma)\in\mathbf{S}^n\even\cup\mathbf{S}^n\odd$.
\end{definition}
\noindent The definition of {\pc} oracles is very similar to that of
{\rpc} oracles in Definition \ref{def:rpc}, except there is no randomness
in {\pc} oracles -- they are unitaries. 
In fact, 
we believe:
\begin{conjecture}\label{conj:nonRandom}
There exists a {\pc} oracle $\orcl$, and a language $L_\orcl$
 which contains those unary strings $1^n$ where
$\mathscr O_n= P_{\sigma}$ with $S_{\textrm{pre}}(\sigma)\in\pmb{S}^n\even$,
such that $L_\orcl\notin \QCMA^\orcl$.
\end{conjecture}

Theorem \ref{thm:qma} applies equally well whether the oracle is
{\pc} or {\rpc}. So why is it harder to prove Conjecture \ref{conj:nonRandom}
than Theorem \ref{thm:main}? The answer is that Recipe \ref{Recipe} is much easier
to use if the optimal witness depends only on a subset of integers. Note {\rpc} oracles have a one-to-one relationship with a subset of integers, and so the optimal witness only depends on that subset. However for {\pc} oracles, the optimal witness might depend on some details of the permutation, which is more challenging to handle.

For convenience, we define the complexity class $\expQCMA$ to be
the analogue of $\QCMA$, in which the quantum verifier is allowed
exponential time and space, but receives a polynomial length
classical witness. While trivially bounded-error quantum exponential time,
$\BQEXP=\expQCMA$, in general the query complexity of a $\expQCMA$ machine
 is not the same
as the query complexity of a $\BQEXP$ machine.

Our proof works as follows.  We first show that if there is a $\QCMA$
machine that decides $L_\orcl$, for all {\rpc}
oracles $\orcl$, then there will be a $\expQCMA$ machine that decides $L_{\widetilde{\orcl}}$ for any {\pc} oracle $\widetilde{\orcl}$, where, {\it crucially}, the optimal witness only depends on the pre-image subset of the permutation implemented by the oracle. Then using 
Recipe \ref{Recipe}, we
show that there is a language $L_{\widetilde{\orcl}}$ for a {\pc} oracle $\widetilde{\orcl}$ such that no $\expQCMA$ machine that can decide the language using an efficient number of queries to $\widetilde{\orcl}$ (with a witness that only depends on the pre-image subset). This implies that there is no
$\QCMA$ machine that solves the {\rpc} oracle problem.

We first prove the reduction from deciding languages on pre-image correct oracles to languages on randomized pre-image correct oracles. 
\begin{restatable}{lemma}{rand}
\label{lemma:random_to_fixed}
Given a {\rpc} oracle $\orcl$, let $1^n\in L_\orcl$ if $\mathscr O_n=\mathscr
P_{\spi(S)}$ with $S\in\pmb{S}^n\even$. Given a {\pc} oracle $\tilde{\orcl}$
let $1^n\in L_{\widetilde{\orcl}}$ if $\sop O_n=\sop P_{\sigma}$ with
$\sperm(\sigma)\in\mathbf{S}^n\even$. Then if there is a $\QCMA$ machine $M$
that decides $L_\orcl$ for every {\rpc} $\orcl$, then there is a $\expQCMA$
machine $\widetilde{M}$ that decides $L_{\widetilde \orcl}$ for every {\pc} ${\widetilde
\orcl}$ such that $\widetilde{M}$ uses at most a polynomial number of queries to
$\widetilde \orcl$, and on input $1^n$ takes as input a classical witness $w$ that
depends only on $\sperm(\sigma)$.
\end{restatable}

\noindent{\it Proof Sketch:} (Full proof in Appendix
\ref{app:rand}.) Given a permutation $\sigma$,
we can obtain all permutations $\sigma'$ such that
$\sperm(\sigma')=\sperm(\sigma)$ by first applying $\sigma$, and then permuting the
first $N$ elements and the last $N^2-N$ elements separately. Consider a controlled-unitary that, if system $A$ is in state
$\ket{i}$, implements the $i\tth$ in-place
 permutation of the first $N$ and last $N^2-N$ elements  on
system $B.$ If we start with system $A$ in an equal superposition,
apply $\sop P_\sigma$ to $B$, apply the control to $A$ and $B$, and
then trace out system $A$, the result is $\mathscr P_{\sperm(\sigma)}$ 
on system  $A$. Thus, given  any 
{\pc} oracle $\sop P_\sigma$, we can simulate the 
{\rpc} oracle $\mathscr P_{\sperm(\sigma)}.$

Using this simulation trick, we can create an algorithm $\tilde{M}$ using 
a {\pc} oracle that has the same outcomes as any algorithm $M$
using a {\rpc} oracle, which uses the oracle the same
number of times, and has a witness that only depends on
the preimage subset. However, we do not believe the control permutation
can be implemented efficiently, and that is why we must consider
the class $\expQCMA$.

\begin{lemma}\label{lemma:no_lang}
There exists a {\pc} $\orcl$ such that there
is no $\expQCMA^\orcl$ machine  $M$ that decides $L_\orcl$ using a polynomial number of queries,
 where the classical witness on input $1^n$
depends only on $\sperm(\sigma)$, when $\sop O_n=\sop P_{\sigma}$.
\end{lemma}

Note that Lemma \ref{lemma:no_lang}, combined with the contrapositive
of Lemma \ref{lemma:random_to_fixed}, proves Theorem \ref{thm:main}.

To prove Lemma \ref{lemma:no_lang}, we use Recipe \ref{Recipe}. Even though we do not have a true subset-based oracle (the function $g(\sop P_{\sigma})=\sperm(\sigma)$ is not injective), using the constraint that the classical witness depends only on $\sperm(\sigma)$, we can apply the recipe. 

Additionally, while Recipe \ref{Recipe} refers to the 
class $\QCMA$, because we are only making a statement about query complexity (and say nothing about space or time complexity), 
the approach also applies to the query complexity
of $\expQCMA.$ 

We prove steps 2 and 3 of Recipe \ref{Recipe} in
Lemmas \ref{lemm:fixing} and \ref{lemm:X_conditions}. These proofs
are quite similar to the proofs of Lemmas \ref{lemm:fixing_subset}
and \ref{lemm:X_conditions_subset}; the full proofs can be found in
Appendix \ref{app:perm_proofs}.

\begin{restatable}{lemma}{fixing}
\label{lemm:fixing}
Let $0<\alpha<1/2$ be a constant and $p(\cdot)$ be a polynomial
function.  Then there exists a positive integer $n^*(p,\alpha)$,
 such that for every $n>n^*(p,\alpha)$, and every subset family
$\mathbf{S}\subset\mathbf{S}_\emph{even}^n$ such that
$|\mathbf{S}|\geq|\mathbf{S}^n_\emph{even}|2^{-p(n)}$, there exists
a subset family $\mathbf{S}_X\subset \mathbf{S}$ such that $\mathbf{S}_X$ is $\alpha$-distributed. Furthermore the fixed subset $S_\textrm{fixed}$ of $\mathbf{S}_X$ contains at most $N/3$ even elements.
\end{restatable}

\begin{restatable}{lemma}{Xconditions}
\label{lemm:X_conditions}
Let $\mathbf{S}_X$ be the $\alpha$-distributed set created using the Fixing Procedure from Lemma \ref{lemm:fixing}, with fixed subset $S_\textrm{fixed}$. Let $\mathbf{S}_Y=\{S:S\in \mathbf{S}^n\odd, S_\textrm{fixed}\subset S \}$.
Then we can construct an adversary bound to prove that 
for every quantum algorithm $G$, there exists permutations 
$\sigma_x,\sigma_y\in \pmb{\sigma}^n$ with $\sperm(\sigma_x)\in\mathbf S_X$ and
$\sperm(\sigma_y)\in \mathbf{S}_Y$, (that depend on $G$) such that, given oracle access to $\sop P_{\sigma_x}$ 
or $\sop P_{\sigma_y}$, $G$ can not distinguish 
them with probability $\epsilon>.5$ without using $\left(1-2\sqrt{\epsilon(1-\epsilon)}\right)N^{\alpha/2}$ queries.
\end{restatable}

The proof strategy of Lemma \ref{lemm:fixing} (like Lemma
\ref{lemm:fixing_subset}) involves a Fixing Procedure. However, the
details are slightly more complex because we must deal with fixing even and odd
elements.

The proof strategy of Lemma \ref{lemm:X_conditions} is similar
 to Lemma \ref{lemm:X_conditions_subset}, except we use a more
complex relation $\mathbf{R}$ for the adversary bound.
The challenge is that for two similar subsets $S_x$ and $S_y$, 
there exist permutations
$\sigma_x$ and $\sigma_y$ that are extremely dissimilar but for which
$\sperm(\sigma_x)=S_x$ and $\sperm(\sigma_y)=S_y$. We want to create a
relationship $\mathbf{R}$ that connects similar permutations, while only having information
about the structure of the related subsets. To address this problem, we
note that for any two subsets $S_x$ and $S_y$, we can create a one-to-one matching between the elements of $\spi(S_x)$ and the
elements of $\spi(S_y)$ such that each permutation is matched
with a similar permutation. Using this one-to-one matching, we 
 create a relationship $\mathbf{R}$ between permutations that inherits
the properties of the related subsets.

As an immediate corollary of Theorem \ref{thm:qma} and Theorem \ref{thm:main},
there exists a {\rpc} oracle $\orcl$ and language $L_{\orcl}$ such that
$L\notin\QCMA^\orcl$ but $L\in\QMA^\orcl$, and so $\QMA^\orcl\not\subseteq\QCMA^\orcl$.



\section{Acknowledgments}\label{sec:ack}

We are grateful for multiple discussions with Stephen
Jordan regarding permutation oracle verification strategies. 
We appreciate the many people who discussed
this project with us, including Scott Aaronson, David Gosset,
Gus Gutoski, Yi-Kai Liu, Ronald de Wolf, Robin Kothari, Dvir Kafri,
and Chris Umans. BF and SK completed much of this work while at the Joint Center for Quantum Information and Computer Science (QuICS), University of Maryland.

\bibliography{perm_adv}
\bibliographystyle{plainurl}

\appendix

\section{An Adversary Bound for Permutation Oracles}\label{sec:ambainis}

We will prove Lemma \ref{lemm:ambainis}:
\amb*
We note that this is identical to Ambainis' adversary bound for permutations
(see Theorem 6 in \cite{A00}).

\begin{proof}
We assume that we have a control permutation oracle, that acts as
\begin{align}
\sop P\ket{x}_C\ket{i}_A\ket{z}_Q=\ket{x}\ket{\sigma_x(i)}\ket{z}
\end{align}
where the Hilbert space $\sop H_C$ has dimension $|\pmb{\sigma}|$,the Hilbert space $\sop H_A$ has dimension
$V$ and is where the permutation is carried out, and $\sop H_Q$ is a set of 
ancilla qubits. 

Let $\ket{\psi^t}$ be the state of the system immediately after $t$
uses of the control oracle. Let $\ket{\varphi^t}$ be the state of the
system immediately before the $t\tth$ use of the control oracle. Let
$\rho^t$ be the reduced state of the system immediately after $t$ uses
of the control oracle, where systems $A$ and $Q$ have been traced out.
That is, $\rho^t=\tr_{AQ}(\ketbra{\psi^t}{\psi^t})$. Let
$\left(\rho^t\right)_{xy}$ be the $(x,y)^\tth$ element of the density
matrix. Then we will track the progress of the following measure:
\begin{align}
W^t=\sum_{(\sigma_x,\sigma_y)\in R}\left|(\rho^t)_{xy}\right|.
\end{align}
Notice that unitaries that only act on
the subsystems $Q$ and $A$ do not affect $W^t.$

If the state before the first use of the oracle is 
\begin{align}
\ket{\psi^0}=\left(\frac{1}{\sqrt{2|\pmb{\sigma}_X|}}\sum_{\sigma_x\in \pmb{\sigma}_X}\ket{x}+
\frac{1}{\sqrt{2|\pmb{\sigma}_Y|}}\sum_{\sigma_y\in \pmb{\sigma}_Y}\ket{y}\right)
\otimes \ket{\phi}_{AQ},
\end{align}
then following Ambainis (e.g. Theorem 2 \cite{A00}), we have that for an algorithm to succeed with probability
at least $1-\epsilon$ after $T$ uses of the oracle, we must have
\begin{align}
W^0-W^T>\left(1-2\sqrt{\epsilon(1-\epsilon)}\right)\sqrt{mm'}
\end{align}

Now we calculate how much $W^t$ can change between uses of the oracle.
Suppose without loss of generality that
\begin{align}\label{eq:psi_star_def}
\ket{\varphi^t}=\frac{1}{\sqrt{2|\pmb{\sigma}_X|}}\sum_{\sigma_x\in \pmb{\sigma}_X}\sum_{i,z}\alpha_{x,i,z}\ket{x,i,z}_{CAQ}
+\frac{1}{\sqrt{2|\pmb{\sigma}_Y|}}\sum_{\sigma_y\in \pmb{\sigma}_Y}\sum_{i,z}\alpha_{y,i,z}\ket{y,i,z}_{CAQ}.
\end{align}
Then we have
\begin{align}
\ket{\psi^{t}}&=\frac{1}{\sqrt{2|\pmb{\sigma}_X|}}\sum_{\sigma_x\in \pmb{\sigma}_X}\sum_{i,z}\alpha_{x,i,z}\ket{x,\sigma_x(i),z}_{CAQ}
+\frac{1}{\sqrt{2|\pmb{\sigma}_Y|}}\sum_{\sigma_y\in \pmb{\sigma}_Y}\sum_{i,z}\alpha_{y,i,z}\ket{y,\sigma_y(i),z}_{CAQ}\nonumber\\
&=\frac{1}{\sqrt{2|\pmb{\sigma}_X|}}\sum_{\sigma_x\in \pmb{\sigma}_X}\sum_{i,z}\alpha_{x,\sigma_x^{-1}(i),z}\ket{x,i,z}_{CAQ}
+\frac{1}{\sqrt{2|\pmb{\sigma}_Y|}}\sum_{\sigma_y\in \pmb{\sigma}_Y}\sum_{i,z}\alpha_{y,\sigma_y^{-1}(i),z}\ket{y,i,z}_{CAQ}.
\end{align}
Hence for $(\sigma_x,\sigma_y)\in R$, we have
\begin{align}\label{eq:rho_calc}
(\rho^t)_{xy}&=\frac{1}{2\sqrt{|\pmb{\sigma}_X||\pmb{\sigma}_Y|}}\sum_{i,z}\alpha_{x,\sigma_x^{-1}(i),z}\alpha^*_{y,\sigma_y^{-1}(i),z}\nonumber\\
(\rho^{t-1})_{xy}&=\frac{1}{2\sqrt{|\pmb{\sigma}_X||\pmb{\sigma}_Y|}}\sum_{i,z}\alpha_{x,i,z}\alpha^*_{y,i,z},
\end{align}
where $(\cdot)^*$ signifies the complex conjugate.
Now we can calculate $W^t-W^{t-1}:$
\begin{align}\label{eq:Wdiff}
W^{t-1}-W^t=&\sum_{(\sigma_x,\sigma_y)\in R}|(\rho^{t-1})_{xy}|-|(\rho^{t})_{xy}|\nonumber\\
\leq&\sum_{(\sigma_x,\sigma_y)\in R}|(\rho^{t-1})_{xy}-(\rho^{t})_{xy}|.
\end{align}

From Eq. \eq{rho_calc}, we see that whenever
$\sigma_x^{-1}(i)=\sigma_y^{-1}(i)$, we have a cancellation between
the corresponding elements of $(\rho^t)_{xy}$ and
$(\rho^{t-1})_{xy}$.  However, when
$\sigma_x^{-1}(i)\neq\sigma_y^{-1}(i)$, terms do not cancel.
To see this more explicitly, we rewrite Eq. \eq{Wdiff} as
\begin{align}\label{eq:sum_cancel}
W^{t-1}-W^t\leq&\frac{1}{2\sqrt{|\pmb{\sigma}_X||\pmb{\sigma}_Y|}}\sum_{(\sigma_x,\sigma_y)\in R}\left|\sum_z\left[\sum_{i:\sigma_x(i)=\sigma_y(i)}\alpha_{x,i,z}\alpha^*_{y,i,z}+
\sum_{i:\sigma_x(i)\neq\sigma_y(i)}\alpha_{x,i,z}\alpha^*_{y,i,z}\right.\right.\nonumber\\
&\left.\left.
-\sum_{i:\sigma_x^{-1}(i)=\sigma_y^{-1}(i)}\alpha_{x,\sigma_x^{-1}(i),z}\alpha^*_{y,\sigma_y^{-1}(i),z}
-\sum_{i:\sigma_x^{-1}(i)\neq\sigma_y^{-1}(i)}\alpha_{x,\sigma_x^{-1}(i),z}\alpha^*_{y,\sigma_y^{-1}(i),z}\right]\right|.
\end{align}
Consider the sets $T_{x,y}=\{i:\sigma_x(i)=\sigma_y(i)\}$ and
$U_{x,y}=\{\sigma_x^{-1}(i):\sigma_{x}^{-1}(i)=\sigma_{y}^{-1}(i)\}$.
We will show $U_{x,y}=T_{x,y}.$ Suppose $i\in T_{x,y}$. Then
$\sigma_x(i)=\sigma_y(i)=i'$, for some $i'$. But that implies
$\sigma_x^{-1}(i')=\sigma_y^{-1}(i')=i$, so $\sigma_x^{-1}(i')=i\in U_{x,y}$, and thus $T_{x,y}\subset U_{x,y}$. 
The opposite direction is shown similarly. Therefore, those
two sums in Eq. \eq{sum_cancel} cancel, and, moving the summation over $z$ and $i$ outside
the absolute values by the triangle inequality, we are left with
\begin{align}
W^{t-1}-W^t
\leq&\frac{1}{2\sqrt{|\pmb{\sigma}_X||\pmb{\sigma}_Y|}}\sum_{z,(\sigma_x,\sigma_y)\in R}\left(
\sum_{i:\sigma_x(i)\neq\sigma_y(i)}\left|
\alpha_{x,i,z}\alpha^*_{y,i,z}\right|
+\sum_{i:\sigma_x^{-1}(i)\neq\sigma_y^{-1}(i)}
\left|\alpha_{x,\sigma_x^{-1}(i),z}\alpha^*_{y,\sigma_y^{-1}(i),z}\right|\right).
\end{align}

Now we use the AM-GM to bound the terms in the absolute values:
\begin{align}\label{eq:two_sums}
W^{t-1}-W^t\leq&\frac{1}{2}\sum_{z,(\sigma_x,\sigma_y)\in R}\left(\sum_{i:\sigma_x(i)\neq\sigma_y(i)}
\left(
\sqrt{\frac{l_{y,i}}{l_{x,i}}}\frac{|\alpha_{x,i,z}|^2}{2|\pmb{\sigma}_X|}+
\sqrt{\frac{l_{x,i}}{l_{y,i}}}\frac{|\alpha^*_{y,i,z}|^2}{2|\pmb{\sigma}_Y|}\right)\right)\nonumber\\
&+\frac{1}{2}\sum_{z,(\sigma_x,\sigma_y)\in R}\left(
\sum_{i:\sigma_x^{-1}(i)\neq\sigma_y^{-1}(i)}
\left(
\sqrt{\frac{l_{y,i}}{l_{x,i}}}\frac{\left|\alpha_{x,\sigma_x^{-1}(i),z}\right|^2}{2|\pmb{\sigma}_X|}+
\sqrt{\frac{l_{x,i}}{l_{y,i}}}\frac{\left|\alpha^*_{y,\sigma_y^{-1}(i),z}\right|^2}{2|\pmb{\sigma}_Y|}\right)\right)
\end{align}

We now show that for $(\sigma_x,\sigma_y)\in R$, 
\begin{align}
\sum_{i:\sigma_x^{-1}(i)\neq\sigma_y^{-1}(i)}|\alpha_{x,\sigma_x^{-1}(i),z}|^2=&
\sum_{i:\sigma_x(i)\neq\sigma_y(i)}|\alpha_{x,i,z}|^2,\nonumber\\
\sum_{i:\sigma_x^{-1}(i)\neq\sigma_y^{-1}(i)}|\alpha_{y,\sigma_y^{-1}(i),z}|^2=&
\sum_{i:\sigma_x(i)\neq\sigma_y(i)}|\alpha_{y,i,z}|^2.
\end{align}
We prove the first equality, and the second is proven similarly. We define
\begin{align}
 T_{x,y}'&=[V]\setminus T_{x,y},\nonumber\\
 U_{x,y}'&=[V]\setminus U_{x,y}.
 \end{align} 
 Looking at the definition of $T_{x,y}$ and $U_{x,y}$, we see that
 \begin{align}
 T_{x,y}'&=\{i:\sigma_x(i)\neq\sigma_y(i)\}\nonumber\\
 U_{x,y}'&=\{\sigma_x^{-1}(i):\sigma_{x}^{-1}(i)\neq\sigma_{y}^{-1}(i)\}.
 \end{align}
 We previously showed $T_{x,y}=U_{x,y}$, so we have $T_{x,y}'=U_{x,y}'$.
 Therefore
 \begin{align}
\sum_{i:\sigma_x^{-1}(i)\neq\sigma_y^{-1}(i)}|\alpha_{x,\sigma_x^{-1}(i),z}|^2=
\sum_{i:U_{x,y}'}|\alpha_{x,j,z}|^2=
\sum_{i:T_{x,y}'}|\alpha_{x,i,z}|^2
\sum_{i:\sigma_x(i)\neq\sigma_y(i)}|\alpha_{x,i,z}|^2.
\end{align}
Thus, Eq. \eq{two_sums} becomes
\begin{align}
W^{t-1}-W^t\leq&\sum_{z,(\sigma_x,\sigma_y)\in R}\left(\sum_{i:\sigma_x(i)\neq\sigma_y(i)}\left(
\sqrt{\frac{l_{y,i}}{l_{x,i}}}\frac{|\alpha_{x,i,z}|^2}{2|\pmb{\sigma}_X|}+
\sqrt{\frac{l_{x,i}}{l_{y,i}}}\frac{|\alpha^*_{y,i,z}|^2}{2|\pmb{\sigma}_Y|}\right)\right).
\end{align}

Now we switch the order of summation and then use the definition of $l_{x,i}$ and $l_{y,i}$ to get
\begin{align}
W^{t-1}-W^t\leq&\sum_{i\in [V],z}\left(\sum_{(\sigma_x,\sigma_y)\in R:\sigma_x(i)\neq\sigma_y(i)}
\left(
\sqrt{\frac{l_{y,i}}{l_{x,i}}}\frac{|\alpha_{x,i,z}|^2}{2|\pmb{\sigma}_X|}+
\sqrt{\frac{l_{x,i}}{l_{y,i}}}\frac{|\alpha^*_{y,i,z}|^2}{2|\pmb{\sigma}_Y|}\right)\right)\nonumber\\
\leq&\sum_{i\in [V],z}
\left(\sum_{\sigma_x\in \pmb{\sigma}_X}
\sqrt{l_{x,i}\max_{\sigma_y:(\sigma_x,\sigma_y)\in R}l_{y,i}}\frac{|\alpha_{x,i,z}|^2}{2|\pmb{\sigma}_X|}+
\sum_{\sigma_y\in \pmb{\sigma}_Y}\sqrt{l_{y,i}\max_{\sigma_x:(\sigma_x,\sigma_y)\in R}l_{x,i}}\frac{|\alpha^*_{y,i,z}|^2}{2|\pmb{\sigma}_Y|}\right)\nonumber\\
\end{align}
Finally, using the definition of $l_{max}$ we have
\begin{align}
W^{t-1}-W^t\leq&\sqrt{l_{max}}\sum_{i\in [V],z}\left(\sum_{x\in \pmb{\sigma}_X}
\frac{|\alpha_{x,i,z}|^2}{2|\pmb{\sigma}_X|}
+\sum_{\sigma_y\in \pmb{\sigma}_Y}\frac{|\alpha^*_{y,i,z}|^2}{2|\pmb{\sigma}_Y|}\right)\nonumber\\
\leq& \sqrt{l_{max}},
\end{align}
where we have used that Eq. \eq{psi_star_def} is a normalized state.
\end{proof}

\section{Proofs of Lemmas \ref{lemm:fixing_subset} and \ref{lemm:X_conditions_subset}}\label{app:perm_proofs_subset}

\fixingsubset*

\begin{proof}
We prove the existence of $\mathbf{S}'$ by construction. Let
$\mathbf{S}$ be any subset family of $\mathbf{C}(N^2,N)$ such that
$|\mathbf{S}|\geq|\mathbf{C}(N^2,N)|2^{-p(n)}$. We construct
$\mathbf{S}'$ using the following procedure:

\vspace{.5cm}
\fbox {
\parbox{.9\linewidth}{
\noindent {\bf{Fixing Procedure}}
\begin{enumerate}
\item Set $\mathbf{S}'=\mathbf{S}$, and set $S_{\textrm{fixed}}=\emptyset$.
\item
\begin{enumerate}
\item Let $\nu(i)$ be the number of subsets $S\in {\mathbf{S}_X}$ such that
$i\in S$. 
\item If there
exists some element $i$ for which
\begin{align}
 |{\mathbf{S}_X}|>\nu(i)\geq |{\mathbf{S}_X}|N^{-\alpha}
 \end{align}  
 set $\mathbf{S}'\leftarrow\{S:S\in {\mathbf{S}_X}\text{ and } i\in S\},$ set 
 $S_\textrm{fixed}\leftarrow S_\textrm{fixed}\cup i,$ and return to step $2$(a).
 Otherwise exit the Fixing Procedure. 
\end{enumerate} 
\end{enumerate}
}
}
\vspace{.5cm}

The Fixing Procedure by construction will always return a set that 
satisfies Definition \ref{def:distr}. Now we just need to bound the size of $S_{\textrm{fixed}}.$

Let's suppose that at some point in the Fixing Procedure, for sets
$\mathbf{S}'$ and $S_\textrm{fixed},$ we have $.5N$ items fixed.
Suppose for contradiction there is some element $i^*\notin S_\textrm{fixed}$
that appears in greater than $N^{-\alpha}$ fraction of
$S\in\mathbf{S}'$.

Let us look at the set family
$\Sprime=\{S:S\in\mathbf{S}', i^*\in S\}.$
Because $(S_\textrm{fixed}\cup i^*)\subset S$ for all $S\in \Sprime$, there
are $.5N-1$ elements in each $S\in\Sprime$ that can be chosen
freely from the remaining $N^2-.5N-1$ un-fixed elements. Thus, we have
\begin{align}
 |\Sprime|\leq \binom{N^2-.5N-1}{.5N-1}.
 \end{align} 
By assumption
 $|\Sprime|\geq |\mathbf{S}'|N^{-\alpha},$
so 
\begin{align}\label{eq:T_upperbound}
|\mathbf{S}'|&\leq \binom{N^2-.5N-1}{.5N-1}N^{\alpha}\nonumber\\
&\leq \binom{N^2}{.5N}N^{\alpha}\nonumber\\
&\leq(2Ne)^{N/2}N^{\alpha}\nonumber\\
&=2^{N/2(\log(2e)+\log N)+\log(N)\alpha}\nonumber\\
&=2^{O(N)+(N/2)\log N}.
\end{align}

However, we can also bound the size of $\mathbf{S}'$ from the Fixing Procedure. 
Notice that at every step of the Fixing Procedure, the size of $\mathbf{S}'$ is
reduced by at most a factor $N^{-\alpha}$. Since we are assuming $.5N$ elements 
are in $S_\textrm{fixed}$, the Fixing
Procedure can reduce the original set $\mathbf{S}$ by at most a factor
$N^{-\alpha N/2}.$ Since $|\mathbf{S}|\geq|\mathbf{C}(N^2,N)|2^{-p(n)}$,
we have that at this point in the Fixing Procedure
\begin{align}\label{eq:S_lowerbound}
|\mathbf{S}'|&\geq|\mathbf{C}(N^2,N)|2^{-p(n)}N^{-\alpha N/2}\nonumber\\
&=\binom{N^2}{N}2^{-p(n)}N^{-\alpha N/2}\nonumber\\
&\geq N^{N}2^{-p(n)}N^{-\alpha N/2}\nonumber\\
&=2^{N\log N-p(n)-\log(N)\alpha N/2}\nonumber\\
&=2^{-O(N)+N\log N (1-\alpha/2)}.
\end{align}
Notice that as long as $\alpha<1$, for large enough $N$ (in particular,
for $N>2^{n^*}$ for some positive integer $n^*$, where $n^*$ depends on $\alpha$
and $p(\cdot)$), the bound of
Eq. \eq{S_lowerbound} will be larger than the bound of Eq. \eq{T_upperbound},
giving a contradiction. Therefore, our assumption must have been false, and more
than $N/2$ elements can not have been fixed during the Fixing Procedure. Therefore, 
the final set produced by the Fixing Procedure will satisfy point (2) of 
Definition \ref{def:distr}.
\end{proof}

\conditionssubset*

\begin{proof}
Note $\mathbf{S}_Y$ is non-empty, since only $.5N$ elements are in $S_\textrm{fixed}.$

We will use Theorem 6 from \cite{A00}. This result is identical to our Lemma \ref{lemm:ambainis},
except with standard oracles rather than permutation oracles.
We define the relation $\mathbf{R}$ as:
\begin{align}
\mathbf{R}=\{(S_x,S_y):S_x\in\mathbf{S}_X,S_y\in\mathbf{S}_Y\}.
\end{align}
Notice that each $S_x\in \pmb{S}_X$ is paired with every element of
 $\mathbf{S}_Y.$ Thus $m=|\mathbf{S}_Y|.$ Likewise $m'=|\mathbf{S}_X|.$

Now consider $(S_x,S_y)\in \mathbf{R}.$ We first consider the case of
some element $j$ such that $j\in S_x$ but  $j\notin S_y$. By our construction
of $\mathbf S_Y$, $j\notin S_\textrm{fixed}$. We upper
bound $l_{x,j}$, the number of $S_{y'}$ such that
$(S_x,S_{y'})\in \mathbf{R}$ and
$j\notin S_{y'}.$ We use the
trivial upper bound $l_{x,j}\leq|\mathbf{S}_Y|$, which is sufficient
for our purposes.  Next we need to upper bound $l_{y,j},$ the number
of $S_{x'}$ such that $(S_{x'},S_y)\in \mathbf{R}$ and
$j\in S_x.$ Since $S_y$ is paired with every element of $\mathbf S_X$
in $\mathbf{R}$, we just need to determine the number of sets in
$\mathbf S_X$ that contain $j$. Because $\mathbf S_X$
is $\alpha$-distributed, there can be at most $N^{-\alpha}|\mathbf S_X|$
elements of $\mathbf S_X$ that contain $j.$ In this case
we have 
\begin{align}\label{eq:lowerbound_alpha} 
l_{x,j}l_{y,j}\leq|\mathbf{S}_X||\mathbf{S}_Y|N^{-\alpha}. 
\end{align}

We now consider the case that $j\in S_y$ but $j\notin S_x$. (Note this
case only occurs when $S_\textrm{fixed}$ contains less than $0.99N$ elements.) We upper
bound $l_{y,j}$, the number of $S_{x'}$ such that
$(S_{x'},S_y)\in \mathbf{R}$ and
$j\notin S_{x'}.$  Again, we use the trivial
upper bound of $l_{y,j}\leq|\mathbf{S}_X|$, which is sufficient for
our analysis.  Next we upper bound $l_{x,j},$ the number of
$S_{y'}$ such that $(S_{x},S_{y'})\in \mathbf{R}$ and
$j\in S_{y'}.$ In our choice of $\mathbf{R}$, $S_x$ is paired
with every $S_y\in\mathbf S_Y$, so we need to count the
number of $S\in \mathbf S_Y$ that contain $j$. We have
\begin{align}
l_{x,j}=&\binom{N^2-S_\textrm{fixed}-1}{0.99N-S_\textrm{fixed}-1}\nonumber\\
=&\frac{0.99N-S_\textrm{fixed}}{N^2-S_\textrm{fixed}}|\mathbf{S}_Y|\nonumber\\
\leq &\frac{1}{N}|\mathbf{S}_Y|.
\end{align}
Therefore in this case, we have
\begin{align}\label{eq:lowerbound_alpha2}
l_{x,j}l_{y,j}\leq |\mathbf{S}_X||\mathbf{S}_Y|N^{-1}.
\end{align}
Looking at Eq. \eq{lowerbound_alpha} and Eq. \eq{lowerbound_alpha2}, we see that 
because $\alpha<1$, the bound of Eq. \eq{lowerbound_alpha} dominates, and so we have
that
\begin{align}
\sqrt{\frac{mm'}{l_{x,j}l_{y,j}}}\geq 
\sqrt{\frac{|\mathbf{S}_X||\mathbf{S}_Y|}{|\mathbf{S}_X||\mathbf{S}_Y|N^{-\alpha}}}
=N^{\alpha/2}.
\end{align}

Using the contrapositive of Lemma \ref{lemm:ambainis}, if an algorithm $G$ makes
less than $q$ queries to an oracle $\sop F_{S}$ where $S$ is promised
to be in $\mathbf{S}_X$ or $\mathbf{S}_Y$, there exists at least one element of $\mathbf{S}_X$ and
one element of $\mathbf{S}_Y$ such that the probability
of distinguishing between the corresponding oracles less than is $1/2+\epsilon$, where
\begin{align}
\frac{1}{2}\sqrt{2N^{-\alpha/2}q}&>\epsilon.
\end{align}
Equivalently, there exists at least one element of $\pmb{S}_X$ and
one element of $\pmb{S}_Y$ such that in order for $\sop A$ to distinguish 
them with constant bias, one requires $\Omega(N^{\alpha/2})$ queries.
\end{proof}


\section{Proof of Lemma \ref{lemma:random_to_fixed}}\label{app:rand}

\rand*

\begin{proof}
We denote the composition
of two CPTP maps with $\circ$, so $\mathscr E\circ\mathscr F$ means apply
$\mathscr F$ first, and then $\mathscr E.$

For each input $1^n$, $M$ applies an algorithm that takes as input
a standard basis state. Because
$S$ completely characterizes $\mathscr P_{\spi(S)}$, the optimal
witness will depend only on $S.$ 

Suppose on input $1^n$ to $M$, the algorithm is the following:
\begin{align}
\mathscr L_{AB}\circ (\mathscr O)_A\circ (\mathscr U_t)_{AB}\circ
 \cdots\circ (\mathscr U_2)_{AB}\circ (\mathscr O)_A
\circ (\mathscr U_1)_{AB}(\ketbra{w}{w}\otimes\ketbra{\psi_0}{\psi_0})_{AB}
\end{align}
where $\ketbra{w}{w}$ is the witness state (that depends only on $S$)
in the standard basis and $\mathscr U_i$
are fixed unitaries and $\mathscr L$. The two subspaces
$A$ and $B$ refer to the subset where the oracle acts $(A)$ and the rest
of the workspace $(B)$. The two subspaces do {\it{not}}
refer to the tensor product structure of the initial state.

For $i\in [N!(N^2-N)!]$
let $\pmb{\tau}^n=\{\tau_i\}$ be the set of permutations on the elements  of $[N^2]$ 
that do not mix the first $N$ elements with the last $N^2-N$ elements.
Then let $\sop P^{\textrm{C}}_{n}$
be the following control-permutation:
\begin{align}
\sop P^{\textrm{C}}_{n}\ket{i}\ket{j}=
\begin{cases}
\ket{i}\ket{\tau_i(j)} \textrm{ for }i\in [N!(N^2-N)!]\\
\ket{i}\ket{j} \textrm{ otherwise}.
\end{cases}
\end{align}
$\mathscr P ^{\textrm{C}}_{n}$ is the respective CPTP map.

$\sop P^{\textrm{C}}_{n}$ is a completely known unitary that is independent
of the oracle, however, we do not know how to implement this unitary
in polynomial time. This unitary is the reason we consider the class
$\expQCMA$ in this proof rather than the more standard $\QCMA.$ Ultimately,
we care about query complexity - not the complexity of the unitaries
that occur between the oracle applications.

Let
\begin{align}
\ket{\chi_n}=\frac{1}{\sqrt{N!(N^2-N)!}}\sum_{i=1}^{N!(N^2-N)!}\ket{i}
\end{align}

Then on input $1^n$ we have $\tilde{M}$ implement the algorithm
\begin{align}
\mathscr L_{AB}\circ  (\mathscr P^{\textrm{C}}_{n})_{C_tA}&\circ (\mathscr O)_A\circ (\mathscr U_t)_{AB}\circ
\cdots\nonumber\\
&\circ (\mathscr U_2)_{AB}\circ  (\mathscr P^{\textrm{C}}_{n})_{C_1A}\circ (\mathscr O)_A
\circ (\mathscr U_1)_{AB}\left(\ketbra{\chi_n}{\chi_n}^t_C\otimes(\ketbra{w}{w}\otimes\ketbra{\psi_0}{\psi_0})_{AB}\right)
\end{align}
where  $(\mathscr P^{\textrm{C}}_{\tau})_{C_iA}$ means the $C_i\tth$
register controls the $A\tth$ register, and initially, the $C_i\tth$ register
is the $i\tth$ copy of $\ket{\chi_n}$, and $\mathscr O$ is the
CPTP version of the oracle $\sop O$. 

Let $\rho_i(\mathscr O)$ (resp. ${\tilde{\rho}_i(\sop O)}$) be the state of the
system during the algorithm $M$ (resp. $\tilde{M}$) after the
$i\tth $ use of the oracle. Let $\rho_0(\mathscr O)$
(resp. ${\tilde{\rho}_0(\sop O)}$) be the initial state of the respective
algorithms. Then we will show that
\begin{align}
\rho_i({\mathscr P_{\spi(\sperm(\sigma))}})=\tr_{C_1,\dots,C_t}
\left({\tilde{\rho}_i({\sop P_{\sigma}})}\right).
\end{align}
As a consequence of this, the probability distribution of measurement outcome of the
two algorithms will be identical.

We prove this by induction. For the initial step, we have
\begin{align}
\rho_0({\mathscr P_{\spi(\sperm(\sigma))}})=\ketbra{w}{w}\otimes\ketbra{\psi_0}{\psi_0}
\end{align}
while
\begin{align}
\tr_C(\tilde{\rho}_0({\sop P_{\sigma}}))=&
\tr_C\left(\ketbra{\chi_n}{\chi_n}^t_C\otimes(\ketbra{w}{w}\otimes\ketbra{\psi_0}{\psi_0})_{AB}\right)\nonumber\\
=&\ketbra{w}{w}\otimes\ketbra{\psi_0}{\psi_0}.
\end{align}
For the induction step, we need to show
\begin{align}
\rho_k({\mathscr P_{\spi(\sperm(\sigma))}})=\tr_{C_1,\dots,C_t}
\left({\tilde{\rho}_k({\sop P_{\sigma}}})\right).
\end{align}
Because $\mathscr P_{\spi(\sperm(\sigma))}$ has an equal probability of applying $\sop P_\sigma$ for
each $\sigma$ such that $S(\sigma)=S$, we have 
\begin{align}
\rho_k({\mathscr P_{\spi(\sperm(\sigma))}})=&\mathscr P_{\spi(\sperm(\sigma))}
\left(\sop U_k\rho_{k-1}({\mathscr P_{\spi(\sperm(\sigma))}})\sop U_k^\dagger\right)\nonumber\\
=&\frac{1}{N!(N^2-N)!}\sum_{i=1}^{N!(N^2-N)!}\sop P_{\tau_i}
\sop P_\sigma\sop U_k\rho_{k-1}({\mathscr P_{\spi(\sperm(\sigma))}})
\sop U_k^\dagger\sop P_\sigma^\dagger
\sop P_{\tau_i}^\dagger.
\end{align}

On the other hand
\begin{align}
\tr_C\left(\tilde{\rho}_k({\sop P_{\sigma}})\right)=&\tr_C\left((\sop P^\textrm{C}_{\tau})_{C_kA}
\sop P_\sigma\sop U_k(\tilde{\rho}_{k-1}({\sop P_{\sigma}}))
\sop U_k^\dagger\sop P_\sigma^\dagger
(\sop P^\textrm{C}_{\tau})^\dagger_{C_kA}\right)\nonumber\\
=&\frac{1}{N!(N^2-N)!}\sum_{i=1}^{N!(N^2-N)!}\sop P_{\tau_i}
\sop P_\sigma\sop U_k\tr_C\left(\tilde{\rho}_{k-1}({\sop P_{\sigma}})\right)\sop U_k^\dagger\sop P_\sigma^\dagger
\sop P_{\tau_i}^\dagger
\end{align}

Now we need to show $\tilde{M}$ decides $L_\orcl$ for
a {\pc} oracle $\orcl.$ Let's consider
an input $1^n$. Suppose $\sop O_n=\sop P_\sigma$, where
$\sperm(\sigma)\in\mathbf{S}^n\even.$ Then because $M$ decides
$L_\orcl$ for any {\rpc}, there exists 
a witness $w$ that depends only on $\sperm(\sigma)$ such that when
the oracle is $\mathscr P_{\spi(\sperm(\sigma))}$
the output of $M$ is 1 with probability at least 2/3. Using
the same witness $w$, $\tilde{M}$ will therefore
produce output 1 with probability at least 2/3. 

Now consider an input $1^n$ such that $\sop O_n=\sop P_\sigma$
where $\sperm(\sigma)\in\mathbf{S}^n\odd.$ Because $M$ decides $L_\orcl$
for a {\rpc} oracle $\orcl$, when $M$ is run
with the oracle $\mathscr P_{\spi(\sperm(\sigma))}$, for any witness $w$,
$M$ will output 1 with probability at most 1/3. But because $\tilde{M}$
will have the same probability distribution of outcomes, this means that for any 
witness $w$ to $\tilde{M}$, with oracle $\sop P_\sigma$, $\tilde{M}$ will
output 1 with probability at most 1/3.
\end{proof}


\section{Proofs of Lemmas \ref{lemm:fixing} and \ref{lemm:X_conditions}}\label{app:perm_proofs}

\fixing* 

\begin{proof}
We prove the existence of $\mathbf{S}'$ by construction. Let
$\mathbf{S}$ be any subset of $\mathbf{S}_\textrm{even}^n$ such that
$|\mathbf{S}|\geq|\mathbf{S}^n_\textrm{even}|2^{-p(n)}$. We construct
$\mathbf{S}_X$ using the following procedure, which we call the fixing
procedure.

\vspace{.5cm}
\fbox {
\centering
\parbox{.9\linewidth}{
\noindent {\bf{Fixing Procedure}}
\begin{enumerate}
\item Set $\mathbf{S}_X=\mathbf{S}$, and set $S_{\textrm{fixed}}=\emptyset$.
\item
\begin{enumerate}
\item Let $\nu(i)$ be the number of subsets $S\in {\mathbf{S}_X}$ such that
$i\in S$. 
\item If there
exists some element $i$ for which
\begin{align}
 |{\mathbf{S}_X}|>\nu(i)\geq |{\mathbf{S}_X}|N^{-\alpha}
 \end{align}  
 set $\mathbf{S}_X\leftarrow\{S:S\in {\mathbf{S}_X}\text{ and } i\in S\},$ set 
 $S_\textrm{fixed}\leftarrow S_\textrm{fixed}\cup i,$ and return to step $2$(a).
 Otherwise exit the Fixing Procedure. 
\end{enumerate} 
\end{enumerate}
}
}
\vspace{.5cm}

By construction, $\mathbf S_X$ will satisfy Definition \ref{def:distr}.
Now we need to check
that the Fixing Procedure stops before fixing more than $N/3$ even items.

Let's suppose that at some point in the Fixing Procedure, for sets $\mathbf{S}_X$
and $S_\textrm{fixed},$ we have $N/3$ even items fixed. Suppose for contradiction, that at this point,
there is some even element $i^*$ such that $i^*$ appears in greater than $N^{-\alpha}$
fraction of $S\in\mathbf{S}_X$. Let's also assume without loss of generality that
$|S_\textrm{fixed}\cap \mathbb Z\odd|=k\odd\leq N/3$.  

Let us look at the set
\begin{align}
\Sprime=\{S:S\in\mathbf{S}_X, i^*\in S\}.
\end{align}
Because $(S_\textrm{fixed}\cup i^*)\subset S$ for all $S\in \Sprime$, there
are $N/3-1$ even elements that can be freely chosen for $S\in\Sprime$
and $N/3-k\odd$ odd elements that can be freely chosen. Thus, we have
\begin{align}
 |\Sprime|\leq \binom{N^2/2-N/3-1}{N/3-1}\binom{N^2/2-k\odd}{N/3-k\odd}.
 \end{align} 
By assumption
\begin{align}
 |\Sprime|\geq |\mathbf{S}_X|N^{-\alpha},
\end{align}
so 
\begin{align}\label{eq:S_upperbound_parity}
|\mathbf{S}_X|&\leq \binom{N^2/2-N/3-1}{N/3-1}\binom{N^2/2-k\odd}{N/3-k\odd}N^{\alpha}\nonumber\\
&\leq \binom{N^2/2}{N/3}\binom{N^2/2}{N/3}N^{\alpha}\nonumber\\
&\leq(3Ne/2)^{2N/3}N^{\alpha}\nonumber\\
&=2^{2N/3(\log(3e/2)+\log N)+\log(N)\alpha}\nonumber\\
&=2^{O(N)+(2N/3)\log N}.
\end{align}

However, we can also bound the size of $\mathbf{S}_X$ from the Fixing Procedure. 
Notice that at every step of the Fixing Procedure, the size of $\mathbf{S}_X$ is
reduced by at most a factor $N^{-\alpha}$. Since we are assuming $N/3$ even elements 
are in $S_\textrm{fixed}$
and $k\odd\leq N/3$ odd elements are in $S_\textrm{fixed}$, the Fixing
Procedure can reduce the original set family $\mathbf{S}$ by at most a factor
$N^{-\alpha(2N/3)}.$ Since $|\mathbf{S}|\geq|\mathbf{S}^n\even|2^{-p(n)}$,
we have that at this point in the Fixing Procedure
\begin{align}\label{eq:S_lowerbound_parity}
|\mathbf{S}_X|&\geq|\mathbf{S}^n\even|2^{-p(n)}N^{-\alpha(2N/3)}\nonumber\\
&=\binom{N^2/2}{2N/3}\binom{N^2/2}{N/3}2^{-p(n)}N^{-\alpha(2N/3)}\nonumber\\
&\leq(3N/4)^{2N/3}(3N/4)^{N/3}2^{-p(n)}N^{-\alpha(2N/3)}\nonumber\\
&=2^{2N/3(\log(3/4)+\log N)+N/3(\log(3/4)+\log N)-p(n)-\log(N)\alpha 2N/3}\nonumber\\
&=2^{N\log N (1-2\alpha/3)+N\log (3/4)-p(\log(N))}\nonumber\\
&=2^{-O(N)+N\log N (1-2\alpha/3)}.
\end{align}
Notice that as long as $\alpha<1/2$, for large enough $N$ (in particular,
for $N>2^{n^*}$ for some positive integer $n^*$, where $n^*$ depends on $\alpha$
and $p(\cdot)$), the bound of
Eq. \eq{S_lowerbound_parity} will be larger than the bound of Eq. \eq{S_upperbound_parity},
giving a contradiction. Therefore, our assumption, must have been false, and
at this point in the Fixing Procedure, all even elements $i\in[N^2]/S_\textrm{fixed}$ will
appear in at most a fraction $N^{-\alpha}$ of $S\in\mathbf{S}_X$. Thus at the next step of the 
Fixing Procedure, an even element will not be added to $S_\textrm{fixed}$, and the number
of even elements in $S_\textrm{fixed}$ will stay bounded by $N/3.$ The same logic can be reapplied
at future steps of the Fixing Procedure, even if additional odd items are added.
\end{proof}

\Xconditions*

\begin{proof}
Since $\mathbf{S}_X$ is $\alpha$-distributed, there exists a set $S_\textrm{fixed}$
of elements such that $S_\textrm{fixed}\subset S$ for all $S\in\mathbf{S}_X$,
 where $S_\textrm{fixed}$ contains at most $N/3$
odd elements and at most $N/3$ even elements. (Otherwise Condition (2) 
of Definition \ref{def:distr} will not be satisfied.)
We choose
\begin{align}
\pmb{\sigma}_Y&=\{\sigma:\sperm(\sigma)\in\mathbf{S}_Y\},\nonumber\\
\pmb{\sigma}_X&=\{\sigma:\sperm(\sigma)\in \mathbf{S}_X\}.
\end{align} 

We now define the relation $\mathbf{R}$ needed to apply our adversary bound.
For each $(S_x,S_y)\in \mathbf{S}_X\times\mathbf{S}_Y$, we will create
a one-to-one matching in $\mathbf{R}$ between the elements of $\spi(S_x)$ and
$\spi(S_y)$. We first choose any  element $\sigma_x^*\in
\spi(S_x).$ Then we choose a permutation
$\sigma_y^*\in\spi(S_y)$ such that
\begin{itemize}
\item $\forall j\in (S_x\cap S_y), \sigma_x^*(j)=\sigma_y^*(j)$, 
\item $\forall j\in [N^2]\setminus(S_x\cup S_y), \sigma_x^*(j)=\sigma_y^*(j),$
\item $\forall j\in S_x\setminus(S_x\cap S_y), \exists i\in S_y\setminus(S_x\cap S_y)
\textrm{ such that } \sigma_x^*(j)=\sigma_y^*(i) \textrm{ and }
\sigma_x^*(i)=\sigma_y^*(j)$.
\end{itemize}
Since every permutation corresponding to $S_y$ is in $\spi(S_y)$, there
will always be such a $\sigma_y^*$ that satisfies the above criterion.
We choose $(\sigma_x^*,\sigma_y^*)\in \mathbf{R}.$

For $i\in [N!(N^2-N)!]$
let $\pmb{\tau}^n=\{\tau_i\}$ be the set of permutations on the elements of $[N^2]$ 
that do not mix the first $N$ elements with the last $N^2-N$ elements.
By $\sigma_a\circ\sigma_b$, we mean apply first permutation $\sigma_b$,
and  then permutation  $\sigma_a$.  Notice that
\begin{align}
\spi(S_x)=&\{\tau\circ\sigma_x^*:\tau\in \pmb{\tau}^n\}\nonumber\\
\spi(S_y)=&\{\tau\circ\sigma_y^*:\tau\in \pmb{\tau}^n\}.
\end{align}
Furthermore given $\tau\in \pmb{\tau}^n,$ we have
\begin{itemize}
\item $\forall j\in (S_x\cap S_y), \tau\circ \sigma_x^*(j)=\tau\circ \sigma_y^*(j)$, 
\item $\forall j\in [N^2]\setminus(S_x\cup S_y), \tau\circ \sigma_x^*(j)=\tau\circ \sigma_y^*(j),$
\item $\forall j\in S_x\setminus(S_x\cap S_y), \exists i\in S_y\setminus(S_x\cap S_y)
\textrm{ such that } \tau\circ \sigma_x^*(j)=\tau\circ \sigma_y^*(i) \textrm{ and }
\sigma\circ \sigma_x^*(i)=\sigma\circ \sigma_y^*(j)$.
\end{itemize}

For every $\tau\in \pmb{\tau}^n$, we set 
$(\tau\circ\sigma_x^*,\tau\circ\sigma_y^*)\in \mathbf{R}.$ In doing so, we create a
one-to-one correspondance in $\mathbf{R}$ between the elements of $\spi(S_x)$
and $\spi(S_y).$ We then repeat this process for all pairs $(S_x,S_y)\in 
\mathbf{S}_X\times\mathbf{S}_Y.$
The end result is the $\mathbf{R}$ that we will use.

Now we need to analyze the properties of this $\mathbf{R}.$
Notice that each $\sigma_x\in \pmb{\sigma}_X$ is paired to exactly one element  of $\spi(S_y)$
for each $S_y\in \mathbf{S}_Y.$ Thus $m=|\mathbf{S}_Y|.$ Likewise $m'=|\mathbf{S}_X|.$

Now consider $(\sigma_x,\sigma_y)\in \mathbf{R}.$ We consider some element $j$ such that
$\sigma_x(j)\neq\sigma_y(j)$. We first consider the case that
$j\in S_x$. We upper bound $l_{x,j}$, the number of $\sigma_{y'}$ such that $(\sigma_x,\sigma_{y'})\in \mathbf{R}$
and $\sigma_{y'}(j)\neq\sigma_x(j).$ To simplify analysis, 
we use the simple upper bound $l_{x,j}\leq|\mathbf{S}_Y|$, which is sufficient for our purposes.
 Next we need to upper bound $l_{y,j},$ the number
of $\sigma_{x'}$ such that $(\sigma_{x'},\sigma_y)\in \mathbf{R}$ and 
$\sigma_{x'}(j)\neq\sigma_y(j).$ By our construction 
of $\mathbf{R}$, we have $j\notin S_y$. Also, by construction, if $j\notin S_y$, $\sigma_{x'}(j)\neq\sigma_y(j)$ 
if and only if $j\in S_{x'}$. Since $\sigma_y$ is paired to only one element $\sigma_x$ for each
set $S_x$, $l_{y,j}$ is bounded by the number of sets $S_x\in  \mathbf{S}_X$ such that $j\in S_x$.
Because $\mathbf{S}_X$ is $\alpha$-distributed, 
at most a fraction $N^{-\alpha}$ of the sets of $\mathbf{S}_X$
can contain $j$, so $l_{y,j}\leq |\mathbf{S}_X|N^{-\alpha}$. In this case we have
\begin{align}\label{eq:lowerbound_alpha_parity}
l_{x,j}l_{y,j}\leq |\mathbf{S}_X||\mathbf{S}_Y|N^{-\alpha}.
\end{align}

We now consider the case that
$j\in S_y$. We upper bound $l_{y,j}$, the number of $\sigma_{x'}$ such that $(\sigma_{x'},\sigma_y)\in \mathbf{R}$
and $\sigma_{x'}(j)\neq\sigma_y(j).$ To simplify analysis, 
we use the upper bound of $l_{y,j}\leq|\mathbf{S}_X|$, which is sufficient for our analysis.
 Next we need to upper bound $l_{x,j},$ the number
of $\sigma_{y'}$ such that $(\sigma_{x},\sigma_{y'})\in \mathbf{R}$ and 
$\sigma_{y'}(j)\neq\sigma_x(j).$ By our construction 
of $\mathbf{R}$, we have $j\notin S_x$. Also, by construction, if $j\notin S_x$, $\sigma_{y'}(j)\neq\sigma_x(j)$ 
if and only if $j\in S_{y'}$. Since $\sigma_x$ is paired to only one element $\sigma_x$ for each
set $S_x$, $l_{x,j}$ is bounded by the number of sets $S_y\in  \mathbf{S}_Y$ such that $j\in S_y$.
Suppose $S_\textrm{fixed}$ contains $k\even$ even elements and $k\odd$ odd elements. If $j$ is odd,
we have
\begin{align}
l_{x,j}&=\binom{N^2/2-k\odd-1}{2N/3-k\odd-1}\binom{N^2/2-k\even}{N/3-k\even}\nonumber\\
&\leq \frac{2N/3}{N^2/2-N/3}|\mathbf{S}_Y|,
\end{align}
while if $j$ is even (in that case, we must have $k\even<N/3$), we have
\begin{align}
l_{x,j}&=\binom{N^2/2-k\odd}{2N/3-k\odd}\binom{N^2/2-k\even-1}{N/3-k\even-1}\nonumber\\
&\leq \frac{N/3}{N^2/2-N/3}|\mathbf{S}_Y|,
\end{align}
where we've used that
\begin{align}
|\mathbf{S}_Y|=\binom{N^2/2-k\odd}{2N/3-k\odd}\binom{N^2/2-k\even}{N/3-k\even}.
\end{align}

Therefore in this case, we have
\begin{align}\label{eq:lowerbound_alpha2_parity}
l_{x,j}l_{y,j}= |\mathbf{S}_X||\mathbf{S}_Y|O(N^{-1}).
\end{align}
Looking at Eq. \eq{lowerbound_alpha_parity} and Eq. \eq{lowerbound_alpha2_parity}, we see that 
because $\alpha<1$, the bound of Eq. \eq{lowerbound_alpha_parity} dominates, and so we have
that
\begin{align}
\sqrt{\frac{mm'}{l_{x,j}l_{y,j}}}\geq 
\sqrt{\frac{|\mathbf{S}_X||\mathbf{S}_Y|}{|\mathbf{S}_X||\mathbf{S}_Y|N^{-\alpha}}}
=N^{\alpha/2}.
\end{align}

Using the contrapositive of Lemma \ref{lemm:ambainis}, if an algorithm $G$ makes
less than $q$ queries to an oracle $\sop O_{\sigma_x}$ where $\sigma_x$ is promised
to be in $\pmb{\sigma}_X$ or $\pmb{\sigma}_Y$, there exists at least one element of $\pmb{\sigma}_X$ and
one element of $\pmb{\sigma}_Y$ such that the probability
of distinguishing between the corresponding oracles less than is $1/2+\epsilon$, where
\begin{align}
\frac{1}{2}\sqrt{2N^{-\alpha/2}q}&>\epsilon.
\end{align}
Equivalently, there exists at least one element of $\pmb{\sigma}_X$ and
one element of $\pmb{\sigma}_Y$ such that in order for $\sop A$ to distinguish 
them with constant bias, one requires $\Omega(N^{\alpha/2})$ queries.
\end{proof}

\end{document}